\documentclass[authoryear,12pt]{elsarticle2}

\usepackage{a4wide}

\usepackage{amsmath,amsfonts,amsthm,amssymb}
\usepackage{graphicx}

\newtheorem {thm}     {Theorem}

\newtheorem {rmk}  {Remark}

\newtheorem {assump} {Assumption}
\newtheorem{prop}{Proposition}

\newcommand{\mat}{\mathbf{M}}
\newcommand{\eps}{\varepsilon}
\newcommand{\ini}{{\mathrm{I}}}
\newcommand{\eliminado}[1]{}

\def\R{\mathbb{R}}

\newcommand{\bm}{\mathbf{\mathcal{BM}^+}([0,1])}

\newcommand{\lbm}{L^{\infty}\left([0,\infty);\bm\right)}

\newcommand{\cicic}{C^{\infty}\left(\rb^+;C^{\infty}([0,1])\right)}

\newcommand{\rb}{\mathbb{R}}

\def\d{{\mathrm{d}}}

\def\O{{\mathcal{O}}}
\newcommand{\A}{\mathbb{A}}
\newcommand{\B}{\mathbb{B}}

\def\bydef{:=}

\begin{document}

\begin{frontmatter}

\title{From discrete to continuous evolution models: a unifying approach to drift-diffusion and replicator dynamics}

\author{Fabio A. C. C. Chalub}
\address{Departamento de Matem\'atica and Centro de Matem\'atica
e Aplica\c c\~oes, Universidade Nova de Lisboa, 
Quinta da Torre, 2829-516, Caparica, Portugal.}
\ead{chalub@fct.unl.pt}
\author{Max O. Souza}
\address{Departamento de Matem\'atica Aplicada, Universidade Federal
Fluminense, R. M\'ario Santos Braga, s/n, 22240-920, Niter\'oi, RJ, Brasil.}
\ead{msouza@mat.uff.br}

\begin{abstract}
We study the large population limit of the Moran process, assuming weak-selection,  and for different scalings. Depending on the particular choice of scalings, we obtain a continuous model that may highlight the genetic-drift (neutral evolution) or natural selection; for one precise scaling, both effects are present. For the scalings that take the genetic-drift into account, the continuous model is given by a singular diffusion equation, together with two conservation laws that are already present at the discrete level. For scalings that take into account only natural selection, we obtain a hyperbolic singular equation that embeds the Replicator Dynamics and satisfies only one conservation law. The derivation is made in two steps: a formal one, where the candidate limit model is obtained, and a rigorous one, where convergence of the probability density is proved. Additional results on the fixation probabilities are also presented.
\end{abstract}


\end{frontmatter}


\section{Introduction}

\subsection{Background}

An important effort in contemporary Evolutionary Biology is to understand the plethora of models used to describe the time evolution of a given population, here  to be understood on a broad sense. Different models were built to model different aspects of evolution. Population genetics models, aiming at modeling genetic variability, had a natural start with discrete stochastic models at the individual level as, for instance, the Moran and Wright-Fisher models \cite{Ewens}. On the other hand, ecological and epidemiological models assume focus on macroscopic features of the population as is the case of the SIR model and the Lotka-Volterra  system~\cite{Murray_vol1}. This is also the case of the replicator equation~\cite{HofbauerSigmund} and adaptive dynamics~\cite{WaxmanGavrilets_JEB2005,ChampagnatFerriereBenArous_Selection2001}. Mean field approximation were also used, from the onset, in population genetics, as in the Kimura equation \cite{Kimura,Ewens}. Alternative modelings efforts include evolutionary models through cellular automata~\cite{Edelstein_JTB1993}, graph theory~\cite{OhtsukiNowak_JTB2006}, among others. More recently, stochastic effects were also included at the continuous level, for instance in the replicator equation~\cite{BenaimHofbauer_JBD2008,Imhof_AAP2005}. Evolution of structured populations has led to the study of integro-differential equations~\cite{ChampagnatMeleard_JMB2007,Perthame_Transport}. A review on continuous mutation-selection models can be found in~\cite{Burger_Genetica1998}. Models direct at the study of fixation probabilities can be found in~\cite{Lambert2006,Lessard2005,Lessard2007,Proulx2000}.

The importance of unifying evolutionary models has been already observed in~\cite{PageNowak2002}. The difference between evolutionary game dynamics, which can lead to non-homogeneous distribution~\cite{HofbauerSigmund,JMS} in the long time, and the diffusive models that usually led to pure fixation has been also pointed out by~\cite{ChalubSouza06,ChalubSouza:2008}. A detailed study between the relationship between discrete and continuous models was also carried out in~\cite{ChampagnatFerrierMeleard_TPE2006}.

\subsection{Modeling and techniques}

The focus of this work is to understand the large population limit of discrete processes rigorously and to show the intimate connection between apparently diverse models.

\subsubsection{The discrete process}

We shall use the Moran process as a paradigm~\cite{Moran}. Thus, we consider a population of two types---usually called the mutant and the resident---having a fixed number $N$ of individuals.
At a given time, an individual is chosen to reproduce with probability proportional to a certain function termed the \textit{fitness}, which we shall denote by $\Psi$. The function $\Psi$ is usually taken as function of the frequency of a particular type in the population, since this accounts for most of the biological relevant models, and it is convenient from a mathematical standpoint. The individual then replaces a second individual (possibly the same one), which is chosen with a uniform probability $1/N$. Time is then increased by $\Delta t$, the time-step, and the process is repeated. 

From a theoretical point of view, while the discrete model is well understood through standard Markov Chain theory, the invariants of these models, other than probability conservation,  is usually neglected, and we briefly review its mathematical properties.

The fitness will depend, parametrically, on $N$ and $\Delta t$. When $N$ is large, we need to have small time-steps in order to keep an appropriate balance. In this case, a natural assumption is the so-called 
\textsl{weak selection principle}:
\begin{equation}\label{weak_selection}
\Psi^{(i)}(x,N,\Delta t)\approx 1+\left(\Delta t\right)^\nu\psi^{(i)}(x)+\mathrm{o}\left(\left(\Delta t\right)^{\nu}\right)\ ,
\quad N\to\infty, \Delta t\to 0\ , \nu>0 .
\end{equation}

\subsubsection{Scaling up}

We are interested in studying the so-called thermodynamical limits of the discrete models. While these limits can be seen as mean-field type approximation, they are more general than the typical diffusion approximation by partial differential equations through the central limit theorem and Feymann-Kac formula. They are also somewhat less general than the approach by~\cite{ChampagnatFerrierMeleard_TPE2006}, for some aspects of finite population dynamics.

More precisely, consider a one-parameter family of models $\mathcal{M}_\eps$, $\eps>0$, a certain model $\mathcal{M}_0$ and
\begin{enumerate}
 \item Consider a certain family of initial conditions $f_\eps^\ini$ such that $\lim_{\eps\to 0}f_\eps^\ini=f_0^\ini$
in a certain sense;
\item Evolve through the model $\mathcal{M}_\eps$ the initial condition $f_\eps^\ini$ and through the model $\mathcal{M}_0$ the initial 
condition $f_0^\ini$ until the time $t_0\in(0,\infty]$, obtaining $f_\eps(t)$ and $f_0(t)$ respectively, for $t<t_0$;
\end{enumerate}
If, for every time $t<t_0$, we have that $\lim_{\eps\to 0}f_\eps(t)=f_0(t)$, in the same sense as before,
then we say that the model $\mathcal{M}_\eps$ converge in the limit $\eps\to0$ until time $t_0$ and in the same sense as before to the model $\mathcal{M}_0$.

In the limit of large population, and small time-steps, we define the probability density of the states as:
\begin{equation}\label{limit_def}
 p(x,t)=\lim_{N\to\infty, \Delta t\to 0}\frac{P(x,t,N,\Delta t)}{1/N}\ .
\end{equation}
The function $p$ defined by equation~(\ref{limit_def}) will, in most cases, satisfy a certain differential equation, i.e., the continuous model. The particular model will be different depending on the precise scalings imposed.  More precisely, by \textit{scaling}, we  mean the relationship imposed between $N$ and $\Delta t$ when passing to the limits $N\to\infty$, $\Delta t\to 0$. In this work, we also call ``the thermodynamical limit'' the limits  $N\to\infty$, $\Delta t\to 0$ jointly to a given scaling. By extension, we also call ``the thermodynamical limit'' the limit model. 
Without loss of generality, we may impose scalings given by $\left(\Delta t\right)^{-1}\propto N^{\mu}$, $\mu>0$. Then, we find the following table of possible thermodynamical limits:
\begin{center}
 \begin{tabular}{|c|c|c|}
$\mu$&$\nu$&Large population limit\\
\hline
$0<\mu<\frac{1}{2}$&$\quad$&\text{Undefined}\\
$\mu=\frac{1}{2}$&$\nu<\frac{1}{2}$&\text{Undefined}\\
$\mu=\frac{1}{2}$&$\nu=\frac{1}{2}$&\textrm{Replicator-Diffusion Equation}\\
$\mu=\frac{1}{2}$&$\nu>\frac{1}{2}$&\textrm{Diffusion Equation}\\
$\mu>\frac{1}{2}$&$0<\nu<1-\mu$&\text{Undefined}\\
$\mu>\frac{1}{2}$&$0<\nu=1-\mu$&\text{Replicator Equation}\\
$\mu>\frac{1}{2}$&$\nu>1-\mu$&\text{Trivial}
\end{tabular}
\end{center}
This table is synthesized pictorially in Figure~\ref{fig:munu}.

\begin{figure}
\begin{center}
\includegraphics[scale=.3]{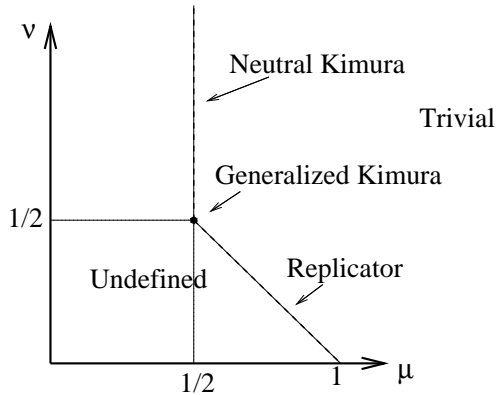}
\end{center}
\caption{Fitnesses scales as $\Psi^{(i)}(x,N,\Delta t)=1+\left(\Delta t\right)^\nu\psi^{(i)}(x)$ and 
the relationship between population size and time-steps is given by $N^{-1}\propto\left(\Delta t\right)^\mu$.
The graph above shows all possible macroscopic limits $\Delta t\to 0$ for the solution of the
Moran process.\label{fig:munu}}
\end{figure}

\subsubsection{Then analysis of continuous models}

The derivation will first be presented in a formal way and then proved rigorous, once a candidate limit is obtained. The techniques used here come from classical numerical analysis~\cite{richtmyer:morton} rather than probability theory---\cite{Ewens,ChampagnatFerrierMeleard_TPE2006}. These ideas assume an earlier analysis of the evolution equation which, in the case of the Replicator-Diffusion model, has been made rigorous in~\cite{ChalubSouza_sub2008}. A skillful, but heuristic, study can be found in~\cite{McKaneWaxman07}.This approach has two clear advantages: it allows for earlier heuristically derivations to be recast in a rigorous framework; it does not depend, necessarily, on the stochasticity of the underlying discrete model. A clear restriction is that it is a mean-field type model and only large populations can be treated.

\subsection{Outline}

In section~\ref{sec:discrete}, we briefly describe the discrete Moran process, with special emphasis on the invariants of the process. In section~\ref{sec:scalings}, we obtain the \textsl{formal} limit
of the Moran process in the large population limit. The precise limit will depend on the different scalings assumption. In particular, for a certain choice, we obtain a \textit{maximal balance} model that takes into account both genetic-drift and natural selection: the replicator-diffusion equation \eqref{generalized_kimura}. In section~\ref{sec:kimura}, we review some of the results obtained in \cite{ChalubSouza_sub2008} regarding the replicator-diffusion equation. As a  particular application of the results, we point out how a generalization of the Kimura equation to the frequency-dependent case can be obtained through duality from \eqref{generalized_kimura}. We also prove the convergence of the discrete model to the continuous one, in the appropriate limit. We also present some additional results on the connection of this model with game-theoretic issues on optimal fixation strategies. In section~\ref{sec:replicator}, we study another class of limit models, and show that they are equivalent to the Replicator equation. Additionally, we also study the fixation probabilities in the small diffusion limit.

\section{Discrete Dynamics}

\label{sec:discrete}

We consider a fixed size population with two types of individuals, $\mathbb{A}$ and $\mathbb{B}$ say. We want to study the following evolutionary process: at a given time, we chose a newborn as copy of one of the individuals taken from the current population with probability proportional to its fitness. This newborn replaces a randomly chosen individual which is, therefore, killed. The process is then repeated  after a timestep of $\Delta t$. This is the so-called Moran process~\cite{Moran}; see Figure~\ref{Moran_fig} for an illustration.

For each type, we associate a fitness function depending on the type frequency:
\[
 \Psi^{(\mathbb{A})},\Psi^{(\mathbb{B})}:[0,1]\to\R^+.
\]

\begin{figure}
\begin{center}
\includegraphics[scale=.15]{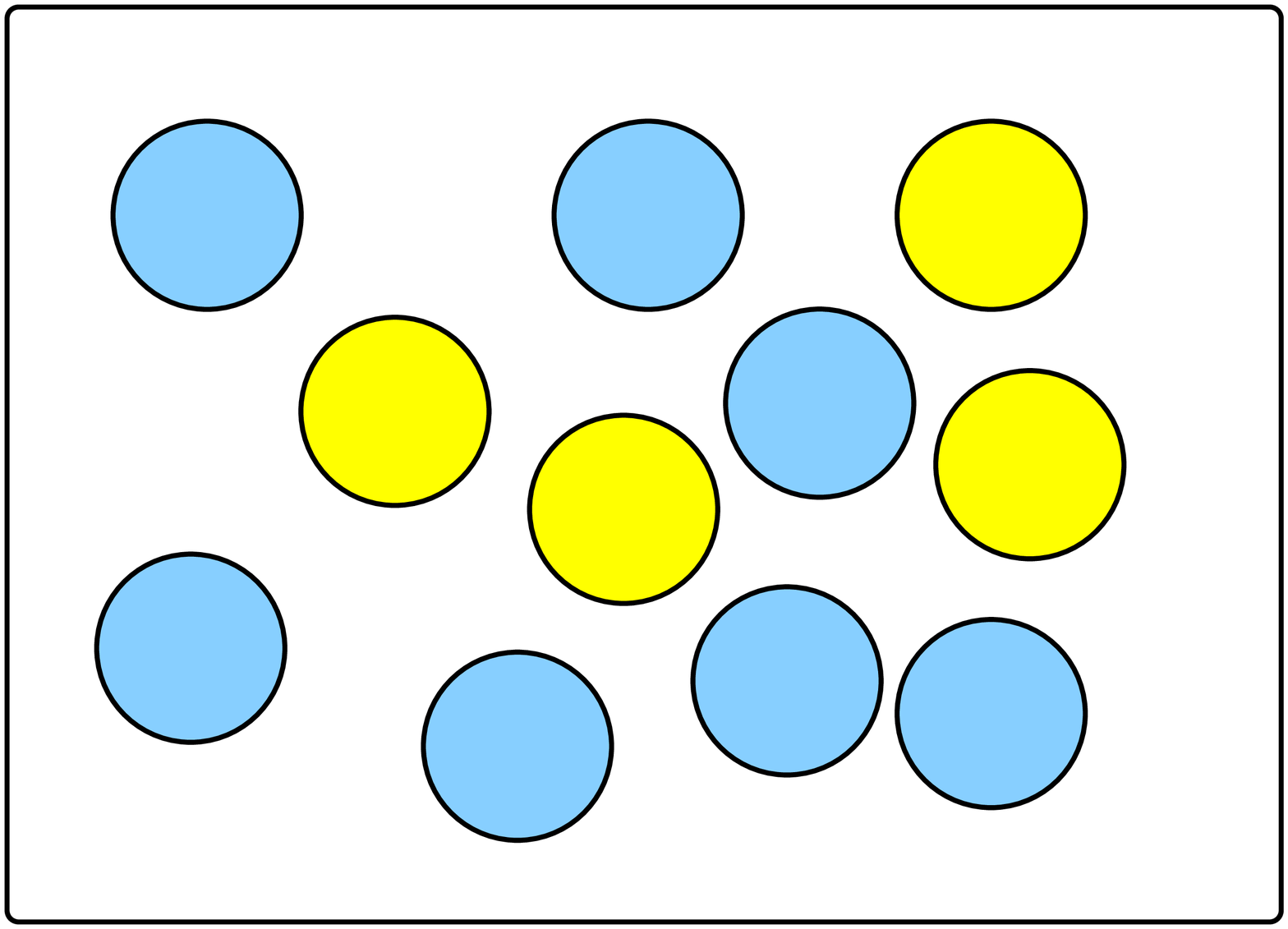}
\hspace{.01\textwidth}\includegraphics[scale=.15]{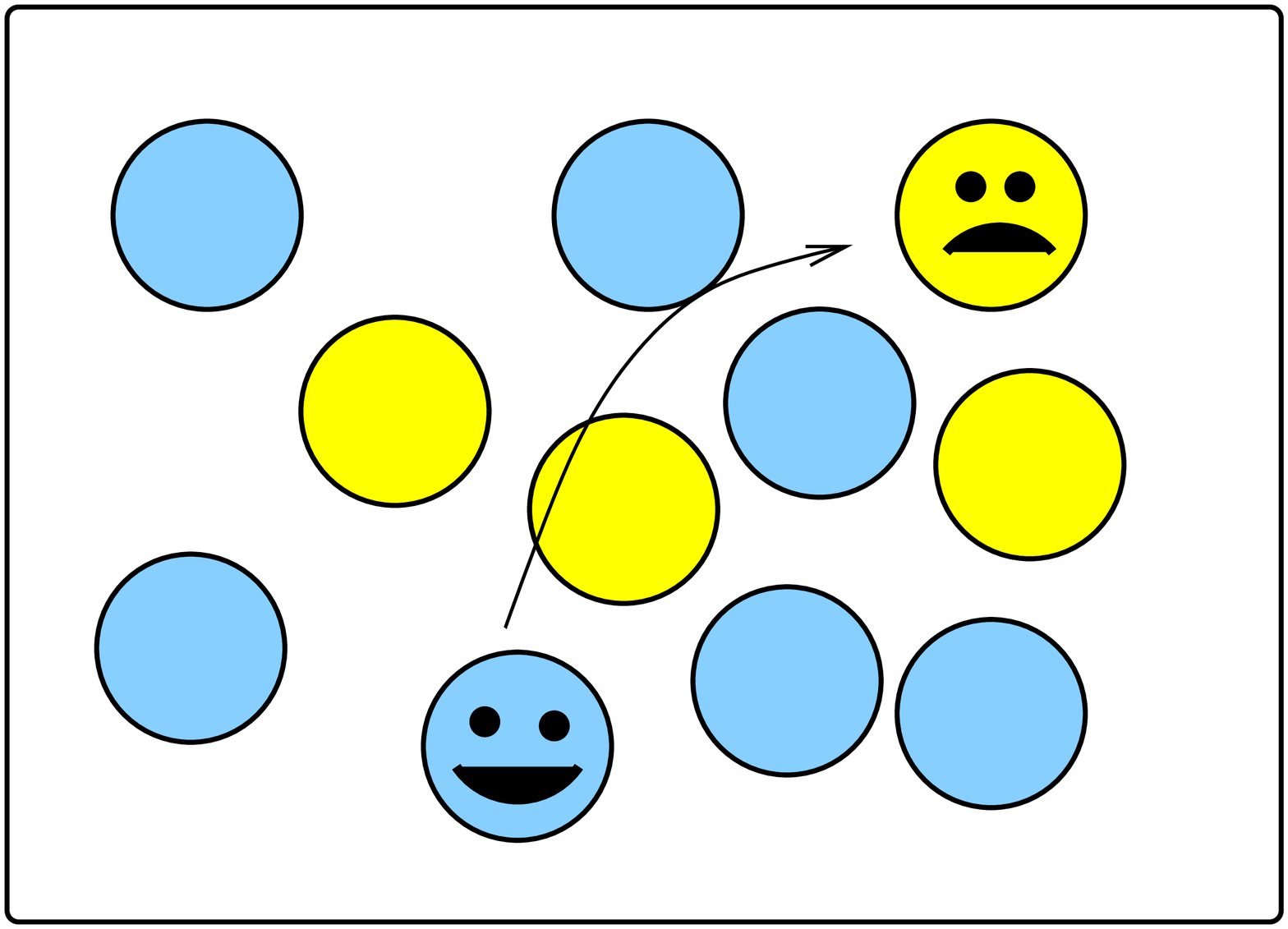}
\hspace{.01\textwidth}\includegraphics[scale=.15]{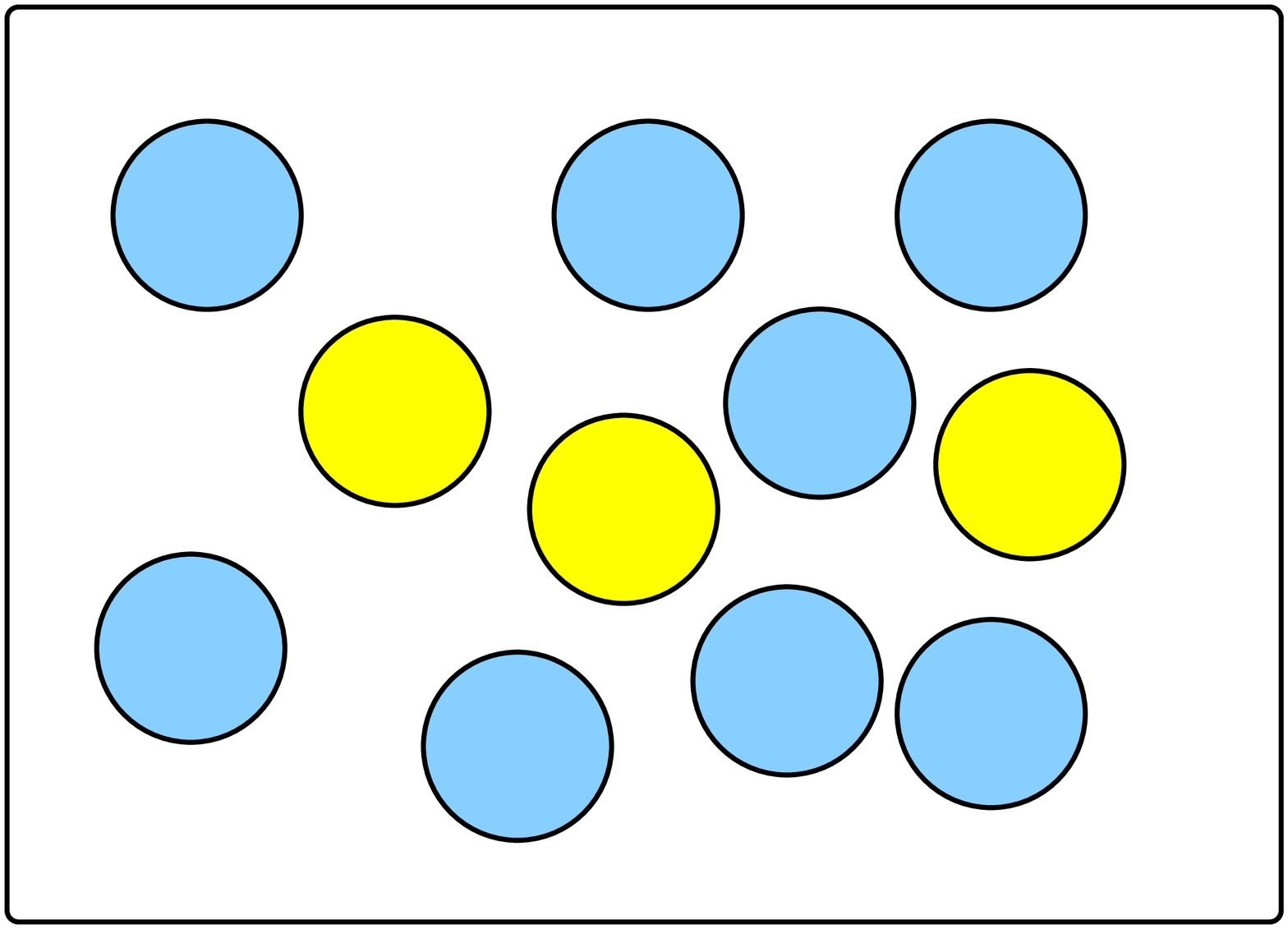}
\vspace{.1cm}
\text{(a)\hspace{.29\textwidth} (b)\hspace{.29\textwidth} (c)}
\end{center}
\caption{The Moran process: from a two-types population~(a) we chose 
one to reproduce with probability proportional to the fitness, ans a second one  
(possibly the same) to be eliminated with equal probability~(b), generating
a new generation~(c).%
\label{Moran_fig}}
\end{figure}
Let $P(x,t,N,\Delta t)$ be the probability that there are $n=xN$,
with $x\in\{0,\frac{1}{N},\frac{2}{N},\cdots,1\}$, individuals of type $\mathbb{A}$ at time $t$, in a population of fixed size $N$. Let $\Delta t$ be the time-step between each re-sampling of the population. 
At each step, the size of the type $\mathbb{A}$ can either increase by one, stay constant or  decrease by one. If we denote the transitions probabilities by $c_+(x,N,\Delta t)$, $c_0(x,N,\Delta t)$ and $c_-(x,N,\Delta t)$, respectively, we find that:
\begin{eqnarray}
\label{c+_def}
c_+(x,N,\Delta t)&=&\frac{x(1-x)\Psi^{(\mathbb{A})}(x,N,\Delta t)}
{x\Psi^{(\mathbb{A})}(x,N,\Delta t)+(1-x)\Psi^{(\mathbb{B})}(x,N,\Delta t)}\ ,\\
\label{c0_def}
c_0(x,N,\Delta t)&=&\frac{x^2\Psi^{(\mathbb{A})}(x,N,\Delta t)}{x\Psi^{(\mathbb{A})(x,N,\Delta t)}+(1-x)\Psi^{(\mathbb{B})}(x,N,\Delta t)}
\\
\nonumber
&&+\frac{(1-x)^2\Psi^{(\mathbb{B})}(x,N,\Delta t)}{x\Psi^{(\mathbb{A})(x,N,\Delta t)}+(1-x)\Psi^{(\mathbb{B})}(x,N,\Delta t)}\ ,\\
\label{c-_def}
c_-(x,N,\Delta t)&=&\frac{x(1-x)\Psi^{(\mathbb{B})}(x,N,\Delta t)}{x\Psi^{(\mathbb{A})}(x,N,\Delta t)+(1-x)\Psi^{(\mathbb{B})}(x,N,\Delta t)}\ .
\end{eqnarray}
From that, we may easily write an equation for the evolution of $P$:
\begin{eqnarray}
\nonumber
P(x,t+\Delta t,N,\Delta t)&=&c_+(x-N^{-1},N,\Delta t)P(x-N^{-1},t,N,\Delta t)\\
\nonumber
&&+c_0(x,N,\Delta t)P(x,t,N,\Delta t)\\
\nonumber
&&+c_-(x+N^{-1},N,\Delta t)P(x+N^{-1},t,N,\Delta t),\\
\label{evol_disc}
&&\quad x\in\{0,1/N,\ldots,1\}\ .
\end{eqnarray}
Let us define the relative fitness
\begin{equation}\label{relative_fitness}
\rho(x,N,\Delta t)=\frac{\Psi^{(\mathbb{A})}(x,N,\Delta t)}{\Psi^{(\mathbb{B})}(x,N,\Delta t)}\ .
\end{equation}
Also, let
\[
f(x)=x(1-x)\quad\text{and}\quad
g(x;\rho)=1+(\rho-1)x\ .
\]
Then it is a straightforward computation to verify that 
\begin{eqnarray*}
c_+(x,N,\Delta t)&=&\frac{f(x)\rho(x,N,\Delta t)}{g(x,\rho(x,N,\Delta t))}\ ,\\
c_0(x,N,\Delta t)&=&1-\frac{f(x)}{g(x,\rho(x,N,\Delta t))}\left(1+\rho(x,N,\Delta t)\right)\ ,\\
c_-(x,N,\Delta t)&=&\frac{f(x)}{g(x,\rho(x,N,\Delta t))}\ ,
\end{eqnarray*}
Let
\[
\mathbf{P}(t)=\left(P(0,t,N,\Delta t),P\left(\frac{1}{N},t,N,\Delta t\right),\ldots,P(1,t,N,\Delta t)\right)^\dagger\ .
\]
Then we can rewrite \eqref{evol_disc} in matrix form as
\[
 \mathbf{P}(t+\Delta t)=\mat\mathbf{P}(t),
\]
where $\mat$ is a $(N+1)\times (N+1)$, tridiagonal matrix, with entries given by
\begin{align*}
&\mat_{ii}=c_0\left(\frac{i}{N},N,\Delta t\right),\quad i=0,\ldots, N\ ,\\
&\mat_{(i+1)i}=c_{+}\left(\frac{i}{N},N,\Delta t\right)\quad i=0,\ldots, N-1,\\
&
\mat_{i(i+1)}=c_{-}\left(\frac{i+1}{N},N,\Delta t\right),\quad  i=0,\ldots, N-1.
\end{align*}
Since $\rho(x,N,\Delta t)>0$, it is easy to see that $\mat$ is
a non-negative matrix. Since $c_{-}(x,N,\Delta t)+c_{0}(x,N,\Delta t)+c_{+}(x,N,\Delta t)=1$,
$\mat$ is column stochastic. 

The long term evolution is described by the following theorem: 
\begin{prop}
\label{Akto0}
Let $\mat$ be as above. Then
\begin{enumerate}
\item
\[
\lim_{k\to\infty}\mat^k=\begin{pmatrix}
1&1-F_1&\ldots&1-F_{N-1}&0\\
0&0&\ldots&\ldots&0\\
\vdots&\vdots&\vdots&\vdots&\vdots\\
0&0&\ldots&\ldots&0\\
0&F_1&\ldots&F_{N-1}&1
\end{pmatrix},
\label{our_prop}
\]
where the $F_n$ satisfy
\begin{align}
&F_{n}=c_{+}\left(\frac{n}{N},N,\Delta t\right)F_{n+1}+c_0\left(\frac{n}{N},N,\Delta t\right)F_n+
c_{-}\left(\frac{n}{N},N,\Delta t\right)F_{n-1}, \nonumber\\
&F_0=0\quad\text{and}\quad F_N=1.
\label{bs:fixeq}
\end{align}
In particular, any stationary state will be concentrated at the endpoints.
\item If $\mathbf{1}$ denotes the vector $(1,1,\ldots,1)^\dagger$,
  $\mathbf{F}=(F_0,F_1,\ldots,F_N)^\dagger$ and  if
  $\langle\cdot,\cdot,\rangle$ denotes the usual inner  product, then 
we have that
\[
\langle \mathbf{P}(t),\mathbf{1} \rangle =\langle
\mathbf{P}(0),\mathbf{1} \rangle 
\quad\text{and}\quad
\langle \mathbf{P(t)}, \mathbf{F} \rangle = \langle \mathbf{P(0)},
\mathbf{F} \rangle.
\]
In particular, the $l^1$-norm of a non-negative initial condition is preserved.
\end{enumerate}
\end{prop}
\begin{proof}
For part 1, it is clear that from the definition of the transition matrix
$\mathbf{M}$ that $\mathbf{e}_0=(1,0,0,\ldots,0)$ and $\mathbf{e}_{N}=(0,0,\ldots,0,1)$
are eigenvectors associated to the eigenvalue $1$. From  the Gersgorin theorem (see~\cite{HJ} for instance),
it is clear that all eigenvalues $\lambda$ are such that $|\lambda|\le 1$. For
tridiagonal matrices with positive entries, the spectrum is real, i.e., 
$\sigma(\mathbf{M})\subset\mathbb{R}$. Let us consider the matrix $\overline{\mathbf{M}}\bydef
\mathbf{M}^N$. Then, $\overline{M}_{00}=\overline{M}_{NN}=1$, 
$\overline{M}_{(i+1)0}=\overline{M}_{iN}=0$ for $i=0,\ldots,N-1$ and
$\overline{M}_{ij}>0$ for $i=0,\ldots,N$ and $j=1,\ldots, N-1$.  This last inequality represents the
fact that it is always possible to go from any non-trivial initial condition 
to any other possible state in $N$ time-steps. In particular,
$0<\sum_{j=1}^{N-1}\overline{M}_{ji}<1$ for $i=1,\ldots, N-1$.

We now prove that any stationary state (i.e. an eigenvector associated to an eigenvalue $\lambda=1$) of the matrix $\overline{\mathbf{M}}$ must be a linear  combination of $\mathbf{e}_0$ and $\mathbf{e}_{N}$. Therefore, this will be also true for $\mathbf{M}$.
Let $\overline{\mathbf{P}}$ be a state with $\overline{P}_0=\overline{P}_{N}=0$. Then
\[
\sum_{j=1}^{N-1}\left(\overline{\mathbf{M}}\,\overline{\mathbf{P}}\right)_{j}=
\sum_{j=1}^{N-1}\sum_{i=0}^{N}\overline{M}_{ji}\overline{P}_i=
\sum_{i,j=1}^{N-1}\overline{M}_{ji}\overline{P}_i<\sum_{i=1}^{N-1}\overline{P}_i\ .
\] 
Finally, $\mathbf{M}\overline{\mathbf{P}}\ne\overline{\mathbf{P}}$.
As for part 2, we first observe that, if a vector $\mathbf{V}$ satisfies
$\mat^\dagger\mathbf{V}=\mathbf{V}$, then we have that
\[\langle \mathbf{P(t+\Delta t)}, \mathbf{V} \rangle=
\langle \mathbf{M}\mathbf{P(t)}, \mathbf{V}
\rangle=
\langle \mathbf{P(t)}, \mathbf{M}^\dagger\mathbf{V}
\rangle=
\langle \mathbf{P(t)}, \mathbf{V}
\rangle.
\]
Hence 
$\langle \mathbf{P}(t),\mathbf{V} \rangle =\langle
\mathbf{P}(0),\mathbf{V}\rangle.$
From the fact that $\mat$ is column stochastic, we easily conclude
that
$\mat^\dagger\mathbf{1}=\mathbf{1}$,
and the first invariant follows. For the second invariant, we observe that
Equation (\ref{bs:fixeq}) can be written in matrix notation as
$\mathbf{M}^{\dagger}\mathbf{F}=\mathbf{F}$,
which concludes the proof.
\end{proof}
\begin{rmk}
The invariants described in part 2, while natural at the discrete level, will play an important role in the correct specification of the continuous problem. 
\end{rmk}

The $F_n$ turns out to be the fixation probability of
type $\mathbb{A}$ individuals, when the process starts in $t=0$ with $n$ type $\mathbb{A}$ 
individuals.

From the definitions of $c_{*}(x,N,\Delta t)$, $*=-,0,+$, and $\rho_n=\rho\left(\frac{n}{N},N,\Delta t\right)$, we see that $F_n$ satisfies:
\begin{equation}
\left\{
\begin{array}{rcl}
&&\rho_nF_{n+1}-\left(\rho_n+1\right)F_n+F_{n-1}=0\ ,\\
&&F_0=0\quad \text{and}\quad F_N=1.
\end{array}
\right.
\label{gen:fixeq}
\end{equation}

Equation (\ref{gen:fixeq}) can be solved exactly, with  solution given  by
\begin{equation}
F_n=\frac{\sum_{k=1}^{n}\prod_{i=1}^{k-1}\frac{1}{\rho_i}}{\sum_{k=1}^{N}\prod_{i=1}^{k-1}\frac{1}{\rho_i}}.
\label{gen:fixsoln}
\end{equation}
The expression given by (\ref{gen:fixsoln}) does not  appear to yield
a simple formula in the general case. For constant fitness, see, however,~\cite{Nowak_EvolutionaryDynamics}.

\section{Scalings and Thermodynamical Limits}
\label{sec:scalings}

\subsection{Preliminaries}
The aim of this section is to derive a continuous
approximation, i.e., a PDE model that approximates the discrete process described in Section~\ref{sec:discrete}, in the large population limit.

We define the probability density that at time $t$ we
have a fraction $x\in[0,1]$ of type $\mathbb{A}$ individuals
\begin{equation*}
\mathcal{P}(x,t,N,\Delta t)\bydef \frac{P(x,t,N,\Delta t)}{1/N}\ ,
\quad\text{with}\quad  x=\frac{n}{N},\quad n=0,1,2,\ldots,N.
\end{equation*}
Furthermore, 
we assume that in the thermodynamical limit, 
$\mathcal{P}(x,t,N,\Delta t)$ converges, in a sense to be made precise in section 4.2, 
  to a function $p(x,t)$ which is sufficiently smooth so that 
\begin{align}
\label{smooth_x}
&p\left(x\pm\frac{1}{N},t\right)=p(x,t)\pm\frac{1}{N}\partial_xp(x,t)+
\frac{1}{2N^2}\partial_x^2p(x,t)+\O(N^{-3})\ ,\\
\label{smooth_t}
&p\left(x,t+\Delta t\right)=p(x,t)+\left(\Delta t\right)\partial_tp(x,t)+\O\left(\left(\Delta t\right)^2\right)\ .
\end{align}

\subsection{Selection scalings}

From the weak selection hypothesis, equation~(\ref{weak_selection}), we 
write the
relative fitness, defined by equation~(\ref{relative_fitness}), asymptotically as
\begin{equation}
 \rho(x,N,\Delta t)\approx 1+
\left(\Delta t\right)^\nu\left(\psi^{(\mathbb{A})}(x)-\psi^{(\mathbb{B})}(x)\right)
+\mathrm{o}\left(\left(\Delta t\right)^{\nu}\right)\ .
\label{rho:ws}
\end{equation}
We also find up to order $y\left(\Delta t\right)^\nu$, defining $y=N^{-1}$,
\begin{align*}
&c_+(x-y,N,\Delta t)\approx x(1-x)\left[1+\left(1-x\right)\left(\Delta t\right)^\nu\left(\psi^{(\mathbb{A})}(x)-\psi^{(\mathbb{B})}(x)\right)\right]
+y(1-2x)\\
&\quad-y^2 -\left(\Delta t\right)^{\nu}y\bigl[x(1-x)^2\partial_x\left(\psi^{(\mathbb{A})}(x)-\psi^{(\mathbb{B})}(x)\right)-
x(1-x)\left(\psi^{(\mathbb{A})}(x)-\psi^{(\mathbb{B})}(x)\right)\\
&\qquad\qquad+(1-2x)(1-x)\left(\psi^{(\mathbb{A})}(x)-\psi^{(\mathbb{B})}(x)\right)\bigr]\ ,\\
&c_0(x,N,\Delta t)\approx 1-2x(1-x)\left[1-\left(\Delta t\right)^\nu\frac{1-2x}{2}
\left(\psi^{(\mathbb{A})}(x)-\psi^{(\mathbb{B})}(x)\right)\right]\\
&c_-(x+y,N,\Delta t)\approx x(1-x)\left[1-x\left(\Delta t\right)^\nu\left(\psi^{(\mathbb{A})}(x)-\psi^{(\mathbb{B})}(x)\right)\right]
+y(1-2x)-y^2\\
&\qquad-\left(\Delta t\right)^{\nu}y\bigr[x^2(1-x)\partial_x\left(\psi^{(\mathbb{A})}(x)-\psi^{(\mathbb{B})}(x)\right)
+x(1-x)\left(\psi^{(\mathbb{A})}(x)-\psi^{(\mathbb{B})}(x)\right)\\
&\qquad\qquad+(1-2x)x\left(\psi^{(\mathbb{A})}(x)-\psi^{(\mathbb{B})}(x)\right)\bigl]\ ,
\end{align*}

Joining these last three equation to the evolution equation~(\ref{evol_disc}) and using assumptions~(\ref{smooth_x})
and~(\ref{smooth_t}), we find
\[
\partial_tp=-\frac{y}{\left(\Delta t\right)^{1-\nu}}\partial_x\left(x(1-x)\left(\psi^{(\mathbb{A})}(x)-\psi^{(\mathbb{B})}(x)\right)p\right)
+\frac{y^2}{2\Delta t}\partial_x^2\left(x(1-x)p\right)\ .
\]
Now we assume $N^{-1}=y=\eps\left(\Delta t\right)^\mu$, with $\eps>0$, and take the limit $\Delta t\to 0$ and, re-scaling time, we find all non-trivial possibilities
according to the Figure~\ref{fig:munu}, where
for $\mu\in\left(\frac{1}{2},1\right]$, $\nu=1-\mu$ we find the 
\textsl{diffusion equation}
\begin{equation}\label{neutral_kimura}
\partial_t p=\frac{1}{2}\partial_x^2\left(x(1-x)p\right)\ ;
\end{equation}
for $\mu=\nu=\frac{1}{2}$, we find the 
\textsl{replicator-diffusion equation}
\begin{equation}\label{generalized_kimura}
\partial_t p=\frac{\eps}{2}\partial_x^2\left(x(1-x)p\right)-\partial_x\left(x(1-x)\left(\psi^{(\mathbb{A})}(x)-\psi^{(\mathbb{B})}(x)\right)p\right)\ ;
\end{equation}
and finally for $\nu>\mu=\frac{1}{2}$ we find the (partial differential version of the) \textsl{replicator equation}:
\begin{equation}\label{pde_replicator}
 \partial_t p=-\partial_x\left(x(1-x)\left(\psi^{(\mathbb{A})}(x)-\psi^{(\mathbb{B})}(x)\right)p\right)\ .
\end{equation}
The invariants in Proposition~\ref{Akto0} become the following conservation laws:
\begin{align}
&\frac{\d}{\d t}\int_0^1p(x,t)\,\d x = 0,\label{cons:prob}\\
&\frac{\d}{\d t}\int_0^1\phi(x)p(x,t)\,\d x = 0,\label{cons:psi}
\end{align}
where  $\phi$ satisfies
\[
 \frac{\eps}{2}\phi''+\left(\psi^{(\mathbb{A})}(x)-\psi^{(\mathbb{B})}(x)\right)\phi'=0,\quad \phi(0)=0,\quad \phi(1)=1\ .
\]
The conservation law \eqref{cons:prob} applies to \eqref{neutral_kimura}, \eqref{generalized_kimura} and \eqref{pde_replicator}. 
The conservation law \eqref{cons:psi} applies to \eqref{neutral_kimura} (with $\phi(x)=x$), and to \eqref{generalized_kimura}.

\section{The Replicator-Diffusion Equation}
\label{sec:kimura}

The equation \eqref{generalized_kimura}, without the conservations laws \eqref{cons:prob} and \eqref{cons:psi}, is a classical parabolic problem that is well posed, without imposing any boundary conditions~\cite{Friedman2,DiBenedetto93}. However, it also true that any such a solution will vanish in the long time limit and, thus, \eqref{cons:prob} and \eqref{cons:psi} cannot be satisfied. 

Is possible to circumvent this difficulty, by using an appropriate notion of weak solution, that is somewhat different from the standard one. This has been done in detail in~\cite{ChalubSouza_sub2008}, to where we refer the interested reader.  In subsection~\ref{subsec:kimura_ss1}, we briefly overview the results presented in~\cite{ChalubSouza_sub2008}; in particular, we point out  that a generalized form of the Kimura equation can be obtained as a dual of \eqref{generalized_kimura} and that, with this notion of solution,  the forward/backward Kolmogorov pair is related through duality. 
 In subsection~\ref{subsec:kimura_ss2} we show that the replicator-diffusion equation~(\ref{generalized_kimura})
is the right approximation of the discrete Moran process in the scaling $\mu=\nu=\frac{1}{2}$, used through this entire
section. Finally, in subsection~\ref{subsec:kimura_ss3}, we deviate  and, on assuming the usual relationship between fitness and payoff in game-theory, we  study the connection of the Replicator-Diffusion fixation probabilities with strategies, and show that the strategies that maximize the fixation probability are exactly the same strategies that maximizes fitness, i.e., the Nash-equilibrium strategies.

\subsection{The solution to the continuous problem}

\label{subsec:kimura_ss1}

The two conservation laws obtained in Proposition~\ref{Akto0} generalizes for the
continuous case at the form $\partial_t\int p\phi\d x=0$, where $\phi$ is any of the
two linearly independent solutions of 
\[
 \frac{\eps}{2}\phi''+\left(\psi^{(\mathbb{A})}(x)-\psi^{(\mathbb{B})}(x)\right)\phi'=0\ .
\]
One solution is given by $\phi_1(x)=1$ and the associated conservation law is simply the
conservation of probability $\int p\d x=1$. The second one is given by
\begin{equation}\label{second_conservation_law}
\phi_2(x)=\frac{\int_0^x\exp\left[-\frac{2}{\eps}\int_0^{x'}\left(\psi^{(\mathbb{A})}(x'')-\psi^{(\mathbb{B})}(x'')\right)\d x''\right]\d x'}
{\int_0^1\exp\left[-\frac{2}{\eps}\int_0^{x'}\left(\psi^{(\mathbb{A})}(x'')-\psi^{(\mathbb{B})}(x'')\right)\d x''\right]\d x'}\ .
\end{equation}
Equation~(\ref{generalized_kimura}) with the two conservation laws is well posed problem, if the appropriate notion of solution is used. See~\cite{ChalubSouza_sub2008} for further details. Here, we only review the main result concerning the solution of this problem. 

Let $\bm$ denote the positive Radon measures in $[0,1]$. Then we have:
\begin{thm}
\label{math:thm}
For a given $p^\ini\in\bm$, there exists a unique (weak) solution $p$ to Equation~(\ref{generalized_kimura}), 
with $p\in\lbm$ and such that $p$ satisfies the conservations laws 
$\partial_t\int_0^1p(x,t)\phi_i(x)\d x$, $i=1,2$. 
The solution can be written as
\[
p(x,t)=r(x,t)+a(t)\delta_0+b(t)\delta_1,
\]
where $r\in\cicic$ is a classical (regular) solution to (\ref{generalized_kimura}) without boundary conditions, and $\delta_y$ denotes the singular measure supported 
at $y$. We also have that 
$a(t)$ and $b(t)$,
belong to $C([0,\infty))\cap C^{\infty}(\rb^+)$. In particular, we have that 
\[
p\in C^{\infty}(\rb^{+},\bm)\cap C^{\infty}(\rb^{+},C^{\infty}((0,1))\ . 
\]
For large time, we have that $\lim_{t\to\infty}r(x,t)=0$, uniformly, 
and that $a(t)$ and $b(t)$, the transient extinction and fixation probabilities,
respectively, are monotonically increasing functions.
Moreover, we have that
\[
\lim_{t\to\infty}p(t,\cdot)=\pi_0[p^\ini]\delta_0+\pi_1[p^\ini]\delta_1,
\] 
with respect to the Radon metric. Finally, the convergence rate is exponential.
\end{thm}
The precise values of $\pi_0$ and $\pi_1$ (respectively, the final extinction and the fixation probabilities)
are given by $\pi_0[p^\ini]=1-\pi_1[p^\ini]$ and $\pi_1[p^\ini]=\int_0^1p^\ini(x)\phi_2(x)\d x$.

Moreover, under this notion of solution, the adjoint equation to \eqref{generalized_kimura} is given by
\begin{align}
\nonumber
&\partial_t f =x(1-x)\partial_x^2 f(x,t) + x(1-x)\left(\psi^{(\mathbb{A})}(x)-\psi^{(\mathbb{B})}(x)\right)\partial_x f(x,t),\\ 
\nonumber
&x\in(0,1),\quad t>0,\\
&f(0,x)=f^0(x).
\label{kimura:adjoint}
\end{align}
Thus, through duality, equation \eqref{kimura:adjoint} generalizes the classical Kimura equation to the frequency selection case. In a particular case, this has been already presented  in \cite{ChalubSouza06}. More recently, using a different approach, Champagnat et al.  have also obtained a generalization of the Kimura equation to frequency dependent selection \cite{ChampagnatFerrierMeleard_TPE2006}.

\subsection{Convergence of the discrete model to equation~(\ref{generalized_kimura}) in the
limit of large populations}
\label{subsec:kimura_ss2}

In this subsection, we rigorously show that: take a general initial condition and evolve
it through the replicator-diffusion equation~(\ref{generalized_kimura}) until time $t$
and obtain the solution $p_{\mathrm{cont}}$; 
consider a discretization in intervals of size $1/N$ of the initial condition and 
evolve it through the Moran process $\frac{t}{\Delta t}=N^2t$ times and obtain
$p_N$. Then the large population limit $N\to\infty$ of $p_N$ is, for any time $t$,
pointwise, the solution of the continuous model $p_{\mathrm{cont}}$. 

This show that the continuous equation~(\ref{generalized_kimura}) is the right \textsl{limit}
of the discrete process (Moran). In a sense, our approach can be described as \textsl{inverse numerical analysis}, as we look for a continuous model that approximates the discrete process, rather than the other way round, which is usually the point of view of numerical analysis. Then, it comes as no surprise that, we are able to prove convergence by using a classical theorem in numerical analysis.
\begin{thm}
\label{thm:conv}
Let $p(x,t,N,\Delta t)$ be the solution of the finite population
dynamics (of population $N$, time step $\Delta t=1/N^2$, i.e., $\mu=\frac{1}{2}$), with initial
conditions given by $p^0(x,N,\Delta t)=p^0(x)$, $x=0,1/N,2/N,\cdots,1$, for
$p^0$ as in theorem~\ref{math:thm}. Assume also the weak-selection limit, i.e, equation~\eqref{rho:ws},
with $\nu=\frac{1}{2}$.
Let $p_{\mathrm{cont}}(x,t)$ be the solution of the continuous model given by equation~(\ref{generalized_kimura}),
with initial condition given by $p^0(x)$.
If we write $p_i^n$ for the $i$-th component of $p(x,t,N,\Delta t)$
in the $n$-th iteration, we have, for any $t^{*}>0$, that
\[
\lim_{N\to\infty}  p^{tN^2}_{xN}
=p_{\mathrm{cont}}(x,t),\quad x\in[0,1],\quad t\in[0,t^{*}].
\]
\end{thm}

\begin{proof}
First, we consider the matrix $\widetilde{\mat}$ obtained from $\mat$ by
deleting the first and last rows and columns.
Then, we observe that the derivation of the thermodynamical limit
shows that the discrete iteration given by $\widetilde{\mat}$ is consistent --- in
the approximation sense~\cite{richtmyer:morton} --- with  
Equation (\ref{generalized_kimura}), without any boundary conditions,
provided that we use \eqref{rho:ws}. From the results of Proposition~\ref{Akto0}, we know that the discrete
iteration is stable, since $\sigma(\widetilde{\mat})\subset(-1,1)$. From
Theorem~\ref{math:thm}, we see that the continuous problem without boundary
conditions is well posed. In this
case, we can then invoke  the Lax-Ricthmyer equivalence theorem
\cite{richtmyer:morton} to guarantee that the discrete model converges to 
the continuum one, in the limit $\Delta t,\Delta x\to0$, with $\Delta
t=(\Delta x)^2$. More precisely, the iteration defined by
$\widetilde{\mat}$ converges to $r(x,t)$, the smooth part of $p_{\mathrm{cont}}(x,t)$;
cf.~\cite{ChalubSouza_sub2008}.

Now returning to the iteration defined by $\mat$. In order to finish
the proof, we only need to show that the first and last rows of $\mat$ converge weakly to the appropriate Dirac masses. We shall do the computation for $x=0$, the case $x=1$ being
similar.

Since $\mathbf{e}_0\in\mathbb{R}^{N+1}$ converges weakly to
$\delta_0$ as $N\to\infty$, we need only to show that it 
has the correct mass at each time $t$.

For $x=0$ the iteration defined by $\mat$ reads
\[
P(0,t+\Delta t,N,\Delta t)=P(0,t,N,\Delta t)+\frac{1}{N}P\left(\frac{1}{N},t,N,\Delta t\right)
\]
Thus, letting $t=0$ and solving the recursion, we have that
\[
P(0,m\Delta t,N,\Delta t)=P(0,0,N,\Delta t)+\frac{1}{N}\sum_{j=1}^{m-1}P\left(\frac{1}{N},j\Delta t,N,\Delta t\right).
\]
Dividing by $1/N$, and taking the limit, we find, in a weak sense, that
\[
\lim_{N\to\infty}p^{tN^2}_0\to\int_0^tr(s,0)\d s+ a_0\delta_0
\]
\end{proof}

\subsection{Strategy dominance}
\label{subsec:kimura_ss3}

For small populations it is know that the strategy that maximizes the fitness do not necessarily is
the same as the strategy that maximizes the fixation probability~\cite{PachecoSantosSouzaSkyrms}.
In this section we identify types and strategies, as defined by game-theory~\cite{HofbauerSigmund,JMS},
and show that, for large populations, any strategy that maximizes the fitness
will also maximize the fixation probability.  For a study for finite populations similar to the one developed here, but that do not extend to infinite populations, see~\cite{NowakSasaki1}.

We say that a strategy $\mathbb{A}$ \textsl{dominates} a strategy $\mathbb{B}$, $\mathbb{A}\succ\mathbb{B}$, if its fixation probability 
$\pi_1[p^\ini]\ge\pi_1^\mathrm{N}[p^\ini]$ for any initial condition $p^\ini$, where $\pi_1^{\mathrm{N}}[p^\ini]=\int_0^1p^\ini(x)x\d x$ 
is the \textsl{neutral} fixation probability. This last value is obtained imposing $\psi^{(\mathbb{A})}=\psi^{(\mathbb{B})}$ in 
equation~(\ref{second_conservation_law}) and then $\phi_2(x)=x$. It is clear that the relation $\succ$ is 
reflexive, anti-symmetric, and transitive; i.e., $\succ$ is a partial order in the set of all possible strategies.
This cannot be made a total order as there are games, like the Hawk-and-Dove game, where no pure strategy
dominates the other.

First note that $\mathbb{A}\succ\mathbb{B}$ if and only
if $\pi_1[\delta_{x^0}]\ge \pi_1^\mathrm{N}[\delta_{x^0}]$ for $x_0\in[0,1]$. Now, let us suppose that $\mathbb{A}$ is the Nash equilibrium
strategy.
Then, for any strategy $\mathbb{B}$ and for any particular composition of the population, $\psi^{(\mathbb{A})}(x)-\psi^{(\mathbb{B})}(x)\ge0$.
This implies that the function $\int_0^x\left(\psi^{(\mathbb{A})}(x')-\psi^{(\mathbb{B})}(x')\right)\d x'$ is a non-decreasing function
of $x\in[0,1]$ and consequently
\begin{align*}
&\frac{1}{x}\int_0^x\exp\left[-\frac{2}{\eps}\int_0^{x'}\left(\psi^{(\mathbb{A})}(x'')-\psi^{(\mathbb{B})}(x'')\right)\d x''\right]\d x'\\
&\quad\ge
\int_0^1\exp\left[-\frac{2}{\eps}\int_0^{x'}\left(\psi^{(\mathbb{A})}(x'')-\psi^{(\mathbb{B})}(x'')\right)\d x''\right]\d x'\ .
\end{align*}
In particular, $\phi_2(x^0)\ge x^0$, $\forall x^0\in[0,1]$ and finally
\[
 \pi_1\left[\delta_{x^0}\right]-\pi_1^\mathrm{N}\left[\delta_{x^0}\right]=\int_0^1\delta_{x^0}(x)\left(\phi_2(x)-x\right)\d x=\phi_2(x^0)-x^0\ge 0\ .
\]

\section{The Replicator Dynamics}
\label{sec:replicator}

In subsection~\ref{subsec:replicator_ss1}, we show how the Replicator Dynamics can be obtained from the discrete Moran process, when the scalings are chosen such that the \eqref{pde_replicator} is the correct limit. In subsection~\ref{subsec:replicator_ss2}, we study the fixation probabilities in the small diffusion limit.

\subsection{The replicator equation as  a large population limit}
\label{subsec:replicator_ss1}

The replicator dynamics models the continuous time evolution of the fraction
of a given type of individuals in a infinite well-mixed population
framework. In its simplest form, it reads as follows:
\begin{equation}\label{replicatorODE}
\dot X=X(1-X)\left(\psi^{(\mathbb{A})}(X)-\psi^{(\mathbb{B})}(X)\right)\ .
\end{equation}
We shall now see that, when the scalings are such that \eqref{pde_replicator} is the correct limit, we actually obtain the Replicator dynamics. We also want to point out that the Replicator dynamics can also be rigorously obtained as the small diffusion limit, or, equivalently, the limit of strong selection and intermediate times, of the replicator-diffusion equation~(\ref{generalized_kimura}). For further details, see~\cite{ChalubSouza:2008pa}.

Let 
\[
 u(x,t)=x(1-x)(\psi^{(\mathbb{A})}(x)-\psi^{(\mathbb{B})}(x))p(x,t)
\]
Then, the equation \eqref{pde_replicator} can be written as
\[
 \partial_tu=x(1-x)(\psi^{(\mathbb{A})}(x)-\psi^{(\mathbb{B})}(x))\partial_xu
\]
Its characteristics~\cite{John} are given by
\begin{align*}
 \frac{\d t}{\d s} &= 1 ,\\
\frac{\d x}{\d s} &= x(1-x)\left(\psi^{(\mathbb{A})}(x)-\psi^{(\mathbb{B})}(x)\right)\\
\frac{\d z}{\d s} &=0
\end{align*}

The projected characteristics in the $x\times t$ plane are, therefore,  given  by
\[
\frac{\d x}{\d t}=x(1-x)\left(\psi^{(\mathbb{A})}(x)-\psi^{(\mathbb{B})}(x)\right)\ ,
\]
which is exactly the equation~(\ref{replicatorODE}).

The solution, for smooth initial condition $p^0(x)$, can be written as
\begin{equation}
p(x,t)=p^0(\Phi_{-t}(x))\frac{\Phi_{-t}(x)(1-\Phi_{-t}(x)))\left(\psi^{(\A)}(\Phi_{-t}(x))-\psi^{(\B)}(\Phi_{-t}(x))\right)}{x(1-x)(\psi^{(\A)}(x)-\psi^{(\B)}(x))}.
 \label{soln:nodiffusion}
\end{equation}
where $\Phi_t(x)$ is the flow map of (\ref{replicatorODE}), and the fraction is interpreted as being one at every equilibrium point of \eqref{replicatorODE}.

Thus, equation \eqref{pde_replicator} can be seen as a Eulerian representation of $u$, while \eqref{replicatorODE} would be a Lagrangian representation. From $u$ to $p$ there is a bijective map. Thus, \eqref{pde_replicator} and \eqref{replicatorODE} are equivalent.

The (\ref{replicatorODE}) can be see as a Lagrangian representation
of $u$, once the initial probability distribution is given, while \eqref{pde_replicator} can be seen as an Eulerian representation of the same quantity.

\begin{rmk}
 A similar version of Theorem~\ref{thm:conv} is valid for the convergence of the Moran process to \eqref{pde_replicator}, in the case of smooth initial conditions. The proof is very similar and, hence, is omitted.
\end{rmk}

\subsection{Fixation probability and the small diffusion limit of equation~(\ref{generalized_kimura})}
\label{subsec:replicator_ss2}

In this subsection we obtain the small diffusion limit $\eps\to 0$ of the
fixation probability $\phi_2$ given by equation~(\ref{second_conservation_law}).
Let us consider an initial condition given by $p^\ini=\delta_{x^0}$ and then the 
fixation probability is given by $\phi_2(x^0)$. Now, we obtain the limit $\eps\to 0$ from equation
(\ref{second_conservation_law}). The only important terms when $\eps\to 0$ come from the minima
of $\int_0^x\left(\psi^{(\mathbb{A})}(x')-\psi^{\mathbb{B})}(x')\right)\d x'$, see~\cite{Hinch1991}.
If the minimum $x^*$ is unique then
\[
 \lim_{\eps\to 0}\phi_2(x)=\chi_{[0,x]}(x^*)\ ,
\]
where $\chi_{I}$ is the characteristic function of the interval $I$. In particular, if $\psi^{(\mathbb{A})}>\psi^{(\mathbb{B})}$
then $x^*=0$ and $\lim_{\eps\to 0}\phi_2(x)=1$, for $x\in(0,1)$, meaning that the type $\mathbb{A}$ will be fixed with probability 1,
if it is initially present, even in a very small quantity.
If there is a unique interior point, $x^*$, such that both pay-off are the same, type $\mathbb{A}$ will be fixed if its
initial presence is superior to $x^*$ and extinct otherwise. Note in figure~\ref{fig:duality} that 
the $\eps$-diffusion is a regularizing effect in the fixation probability. Compare with 
Figure 1 in~\cite{TraulsenNowakPacheco_PRE2006}. See also \cite{ChalubSouza:2008pa}.

\begin{figure}
\begin{center}
\includegraphics[scale=.3]{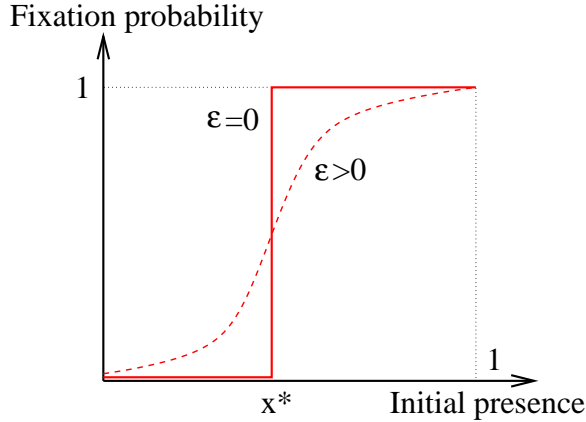}
\end{center}
\caption{Fixation probability $\phi_2(x)$ \textsl{versus} the initial presence $x$ (i.e., initial condition
is given by $p^\ini=\delta_x$ The solid line shows the (discontinuous) $\lim_{\eps\to 0}\phi_2$, when there
is a unique interior point such that $\psi^{(\mathrm{A})}(x^*)=\psi^{(\mathrm{B})}(x^*)$. This is the
case of games like the stag-hunt or the hawk-and-dove. The dashed lines indicate $\phi_2$ for positive
$\eps$. Compare with figure 1 in~\cite{TraulsenNowakPacheco_PRE2006}.\label{fig:duality}}
\end{figure}

If there are more than one minimum for 
$\int_0^x \left(\psi^{(\mathbb{A})}(x')-\psi^{\mathbb{B})}(x')\right)\d x'$, the fixation probability in the
limit $\eps\to 0$ is the average of the function $\chi_{[0,x]}$ over all minima. 
For example, if $\int_0^x\left(\psi^{(\mathbb{A})}(x')-\psi^{\mathbb{B})}(x')\right)\d x'=x(1-x)$,
then $\lim_{\eps\to 0}\phi_2(x)=\frac{1}{2}\left(\chi_{[0,x]}(0)+\chi_{[0,x]}(1)\right)=\frac{1}{2}$, 
for $x\in(0,1)$.
In the neutral case, $\lim_{\eps\to 0}\phi_2(x)=\int_0^1\chi_{[0,x]}(x')\d x'=\int_0^{x}\d x'=x$.

\section{Discussion}

The results presented here can be seen as an alternative approach both to mean-field approximations usually made through diffusion theory~\cite{Ewens} and to more recent probabilistic approaches~\cite{ChampagnatFerrierMeleard_TPE2006}. 

While the analysis presented was focused on the Moran process, the paradigm presented here can accommodate more general discrete models. The formal derivation of the continuous model---in a self-consistent manner---as a first step, should be seen as a feature, since it allows for the possibility that a  heuristically derived approximation can be made rigorous. It is also less dependent on the stochasticity of the subjacent discrete model, which can be an advantage in many situations. 

It also shows that diverse models as evolutionary game dynamics~\cite{HofbauerSigmund,PageNowak2002} and neutral genetic variability~\cite{Kimura,Ewens} can be tighten from the discrete process. As quoted in \cite{Watterson1996}, while Darwinian motto is the survival of the fittest, Kimura work on neutral selection in genetics was described by himself as survival of the luckiest. Here we show that both can be accommodated---see also \cite{ChampagnatFerrierMeleard_TPE2006,ChalubSouza06}. Moreover, we show that from a unique discrete process, we can have both effects, and recover either Darwinian, pure Stochastic, or a balanced evolutionary dynamics. 
A more detailed study of the interconnection between Darwinian and Stochastic evolution as a multiscale process is being carried by~\cite{ChalubSouza:2008pa}. This also suggests that an important issue in Mathematical Evolutionary Dynamics is to understand the interplay between the effects and time scales in the diverse biological phenomena linked to the evolutionary process.

\section*{Acknowledgments}

FACCC is partially supported by FCT/Portugal, grants  
PTDC/MAT/68615/2006 and PTDC/MAT/66426/2006. 
MOS is partially supported by FAPERJ grant \# 170.382/2006. MOS thanks the support and hospitality of FCT/UNL and Complexo Interdisciplinar/UL. Part of this work  has been done during the \textit{Special Semester on Quantitative Biology Analyzed by Mathematics}, organized by RICAM, Austrian Academy of Sciences.

\bibliography{thermo}
\bibliographystyle{plain}

\end{document}